\pdfoutput=1
\documentclass[10pt,twocolumn,twoside]{IEEEtran}
\usepackage{algorithm}
\usepackage{algpseudocode}
\usepackage{amsmath,epsfig,cite}
\usepackage{amsthm}
\usepackage{amssymb}
\usepackage{bm}
\usepackage{caption}
\usepackage{float}
\usepackage{epsfig}
\usepackage{pgfplots}
\usepackage{graphics}
\usepackage{graphicx}
\usepackage{epstopdf}
\usepackage{latexsym}
\usepackage{mathtools}
\usepackage{setspace}
\usepackage{enumitem}
\usepackage{xcolor}
\usepackage[utf8]{inputenc}
\usepackage{fancyhdr}
\newtheorem{definition}{Definition}
\newtheorem{theorem}{Theorem}

\DeclareMathOperator*{\argmax}{arg\,max}

\newcommand*{\permcomb}[4][0mu]{{{}^{#3}\mkern#1#2_{#4}}}
\newcommand*{\perm}[1][-3mu]{\permcomb[#1]{P}}

\makeatletter
\newcommand*{\rom}[1]{\expandafter\@slowromancap\romannumeral #1@}
\makeatother
\begin{document}

\title{Game Theoretic Semi-Distributed D2D Resource Allocation Underlaying an LTE Network}

\author{Anushree Neogi\thanks{This work was done when the author was completing her PhD at Indian Institute of Technology Bombay, Mumbai, Maharashtra, 400076, India (e-mail: anushree.iitg@gmail.com). }}

\fancyhead[L]{This paper is a preprint of a paper submitted to IET Networks. If accepted, the copy of record will be available at the IET Digital Library.}
\renewcommand{\headrulewidth}{0pt}

\maketitle
\thispagestyle{fancy}

\begin{abstract}
To devise D2D resource allocation algorithms in underlay D2D communications the channel state information (CSI) between the D2D transmitters and the BS and the D2D receiver CSI (DCSI-R) needs to be transmitted to the BS. However, this increases the control overhead and power wastage which increases with a fast fading channel since the CSI needs to be transmitted in every time slot. Most of the existing works assume DCSI-R availability at the BS. However, a few works assume its unavailability and determine the Nash equilibrium which may not be Pareto optimal. We address this problem and within a game theoretic framework propose a suboptimal semi-distributed D2D resource allocation algorithm. We consider the channel to exhibit path loss. Our goal is to maximize the social utility of the D2D users while meeting their utility requirements and the signal-to-interference-plus-noise ratio (SINR) requirements of the CUs to reach a Pareto optimal solution. Next, we consider shadowing, fast fading and mobility of CUs and propose another algorithm which is a modification of our first proposed algorithm. Through simulations we observe that the first algorithm does not perform well practically but the second algorithm is very robust to channel randomness and CU mobility.
\end{abstract}

\begin{IEEEkeywords}
Device-to-device (D2D) communication, resource allocation, game theory, social utility, Pareto optimal.
\end{IEEEkeywords}

\section{Introduction}

\subsection{Background}
 Because of the rapid increase in multimedia applications which require higher data rates, the fifth generation (5G) wireless network is being designed to meet this demand. Other requirements of the 5G network are low transmission delay, low energy consumption, high spectrum efficiency and high system capacity. Device-to-device (D2D) communication is therefore a promising component of 5G wireless networks. D2D communications are short range direct communications between D2D users, without the data passing through the base station (BS). This increases the system throughput and reduces the transmission delay and energy consumption of the network.  In D2D communications underlaying a Long Term Evolution (LTE) network, the D2D users are assumed to reuse the resources allocated to the cellular users (CUs)  with the aid of efficient resource allocation algorithms for the D2D users which increases the spectrum efficiency and the system capacity.  In an LTE network, resources are blocks of time and frequency called resource blocks (RBs). The CUs are treated as primary users in the D2D resource allocation algorithms which implies that the CUs' quality of service  (QoS) should not get affected by the interference from the D2D users' communications. Besides, the QoS requirements of the  D2D users need to be guaranteed too. Hence, interference management between the CUs and the D2D users is a technical challenge while designing resource allocation algorithms for the D2D users. Applications of D2D communications for 5G are local data, voice and video transmissions, emergency communications in disaster management, extension of network coverage via relays and  cellular vehicle-to-everything (C-V2X) communications for road safety and regulating traffic \cite{a0}.\\
\indent Currently, at least 13 C-V2X vehicle models have been launched in China \cite{a1}. C-V2X technology was first finalized by the 3GPP in Release 14 of its standards. It allows direct communications in the sub-6 GHz band (5.9 GHz) that has been dedicated  by governments around the world for Intelligent Transport System (ITS). C-V2X includes four types of communications: vehicle-to-vehicle (V2V), vehicle-to-infrastructure (V2I), vehicle-to-pedestrian (V2P) and vehicle-to-network (V2N) which are enabled by a new interface called the PC5 interface, also known as the sidelink \cite{a2}.  The conventional Uu interface enables V2N communications in which a vehicle, road side unit (RSU) or pedestrian can receive information from the network about the road conditions and traffic. The spectrum resources are dedicated for C-V2X communications and are not reused because  C-V2X communications need to be highly reliable and must not be subject to interference.\\
\indent However, the implementation and deployment of D2D communications in the LTE  underlaid operation has been in a prefatory state over the years \cite{a3}. For 6G, D2D applications are support for machine-to-machine (M2M) communications and performance improvement of unmanned aerial vehicles (UAVs) or drones in emergency situations, military operations and remote sensing \cite{a4}. 

\subsection{Motivation}
\indent Most of the previous works on D2D resource allocation underlaying an LTE network assume that the BS as a central resource allocation unit has knowledge of all the channel gains among the different communicating entities from which the BS can determine the achievable rates of all the communication links in order to design a centralized resource allocation algorithm  \cite{a}-\cite{c}. The channel gains between a CU and the BS and a D2D transmitter and the BS are known because every user equipment (UE) transmits its channel quality indicator (CQI) to the BS either periodically or aperiodically \cite{d}. However, it is difficult for the BS to acquire the D2D receiver CQI, i.e. the CQI between a D2D transmitter and a receiver in a pair and that between a CU and a D2D receiver because the cardinality of these channel gains is large. It is equal to the number of CUs times the D2D receivers plus the number of D2D pairs. For a fast fading channel, even if these channel gains are known  through some technique, conveying them to the BS in every time slot requires significant control overhead and power. When these channel gains are not known at the BS, it is a limiting factor to the BS's resource allocation capability.  This makes D2D resource allocation a challenging problem. \\
\indent With this difficulty, D2D resource allocation problems are generally modeled game theoretically and distributed or semi-distributed resource allocation algorithms are designed for each UE to achieve a certain network objective. As compared to a centralized algorithm in which the decisions are taken by the BS, in a distributed algorithm they are taken by the UEs and in a semi-distributed algorithm they are taken by both the UEs and the BS. This reduces the computational load at the BS considerably as some of the computations are done by the UEs themselves. The resulting Nash equilibrium solution sometimes fails to achieve a Pareto optimum solution since the best response of every player may not maximize the social utility,  i.e., the sum of utilities of all the players. A Pareto optimum allocation \cite{e},  \cite{f} is a benevolent allocation in which it is not possible to reallocate resources to make someone better off without making someone worse off and which also maximizes the social utility which a Nash equlibrium may not be able to. The inefficiency of the Nash equilibrium can be quantified through the price of anarchy \cite{g1}.  Some different forms of equilibria like potential function maximizers \cite{g} and coarse correlated equilibria \cite{h},  \cite{i} have also been investigated which gives better behaviour than the Nash equilibrium.

\subsection{Related Work}
We give an overview of some of the research works pertaining to D2D resource allocation in a game-theoretic framework.\\\\
\noindent
\emph{1) Stackelberg game, auction,  graph-theoretical and coalition formation game based  D2D resource allocation algorithms:} In \cite{j}, the authors employ a Stackelberg game theoretic framework for joint power and channel allocation for the D2D users. The Stackelberg model is also applied in \cite{k } in which the authors  propose a suboptimal low complexity two-stage algorithm. In \cite{l}, an iterative scheme of resource allocation is proposed  based on the Stackelberg game model. The study proposed in \cite{m} formulates the D2D resource allocation problem as a reverse iterative combinatorial auction game and an auction algorithm is proposed. However, because of the iterative nature of these algorithms the communication overhead is substantial. The authors of \cite{n} propose a centralized graph-theoretical resource allocation scheme. They model the D2D power control problem as  an exact potential game game with incomplete information and use Q-learning better-reply dynamics to achieve a pure Nash equilibrium which maximizes the potential function. In \cite{o}, two schedulers have been proposed by playing a transferable utility (TU) game and a Stackelberg game among the D2D pairs and the CUs. However, the complexity of the combined scheduler is high. In the TU game, the first scheduler while allocating RBs, does not ensure priority to the CUs as primary users because they consider both CUs and D2D pairs as players. Therefore, the CUs may not get the RBs which maximize their objective function and the CUs' RBs do not get reused by the D2D players. Because of this, the second scheduler in the Stackelberg game allocates the RBs of the CUs that are not reused to those D2D players which have not been allocated any RB. Thus, allocations are done twice which is inefficient. In \cite{p}, the D2D resource allocation problem is formulated as a coalition formation game. To assist the coalition initialization and the coalition generation priority sequences are proposed which increases  the communication overhead. The power allocation problem is solved using the whale optimization algorithm. 
\\\\
\noindent
\emph{2) Reinforcement learning (RL) based D2D resource allocation algorithms:} Recent research efforts employ reinforcement learning techniques in D2D resource allocation. In \cite{q}, a two-step Stackelberg game is proposed to formulate the D2D power and resource allocation problem. In the leader and follower games, determining a suitable price value is difficult. Thus, they use RL for learning the prices. However, the channel model is a path loss model and they have not considered fading.  In \cite{r}, the authors  propose a Stackelberg game guided multi-agent deep reinforcement learning (STDRL) resource allocation algorithm. They derive a  Stackelberg Q-value at each training step which guides the learning direction of the agents in a multi-agent deep reinforcement learning (MADRL) framework. Though, MADRL methods are known to perform well in a dynamic environment, the computational complexity of the algorithm is high because the Stackelberg game and the RL problem are solved alternately. Moreover, in each training step, interactions among the agents are needed for sharing actions which increases the communication overhead. In \cite{s}, a  deep reinforcement learning (DRL) based resource allocation algorithm is proposed. The authors modify the  deep Q-network  (DQN) algorithm which gives exceedingly optimistic estimated values and propose a double deep Q-network (DDQN) algorithm that separates the selection and evaluation of an action by using different neural networks. However, the solution is a Nash equilibrium which may not be Pareto optimal.  However, these RL based  algorithms do not guarantee the rate requirements of the CUs because the D2D players do not check before transmiting if the CU's rate is decreasing below its minimum rate. The rate requirements of the D2D users are also not guaranteed.\\\\
\noindent
\emph{3) Social utility maximization based D2D resource allocation algorithms:} In \cite{e}, the authors propose a Pareto optimal learning algorithm within a game-theoretic framework in which the system states converge probabilistically to stochastically stable states provided interdependence is ascertained in the game. Interdependence implies that the utility of a player is affected by the choice of actions of the other players. In a stochastically stable state, the action profile is a  Pareto (efficient) action profile which maximizes the social utility of the players. The players have no knowledge about the actions and utilities of other players and can observe only their own utilities. This algorithm has been applied to the optimal association problem of CUs to BSs in \cite{t} and for the control of wind turbines in a wind farm \cite{u}. \\
\indent For the D2D resource allocation problem, the algorithm of \cite{e} has been applied in \cite{v} and \cite{w}. In \cite{v}, a power and resource allocation algorithm for D2D users and small cell users is proposed. For resource allocation, the authors have applied the algorithm of \cite{e}. They report that their proposed algorithm achieves a stochastically stable state and it is a local optimal solution. However, their proposed resource allocation algorithm is suboptimal because the stochastically stable states cannot be determined for a Markov process that is not a regular perturbed Markov process  \cite{e}. As they consider the channel to be a fast fading channel, the Markov process is not aperiodic and therefore not ergodic. As a result, the Markov process is not a regular perturbed Markov process. In \cite{w}, a learning algorithm is proposed which the authors report maximize the sum of utilities of the D2D players where the utility of a D2D player is a logarithmic function of its rate. Since they have considered block Rayleigh fading, the algorithm is optimal only when the channel is constant in a block. However, the block needs to be sufficiently large for the algorithm to converge, which is not a practical assumption in a block Rayleigh fading channel.  When the channel changes in the next block, the system state space changes and thus the stochastically stable states change. If the block is not sufficiently large, before the algorithm can converge to the stochastically stable states, the state space and  stochastically stable states would change in the next block.  The authors have also not given the value of the block size parameter in the simulation parameters table.  In both  \cite{v} and  \cite{w}, interdependence in the games are ensured by allowing players to use the same resources because of which interference is caused among the players. This decreases a player's rate and therefore its utility. Because of the actions of other players, a player's utility gets affected which satisfies the interdependence condition in the games. However, as more players are accommodated in the cell, the rate of a player keeps decreasing because of the interference from the other players. If the number of D2D pairs is very high then the average rate of the D2D players decreases.  Therefore, this is a limiting factor to the number of D2D pairs that can be accommodated in the cell. Besides, the algorithms of \cite{v} and  \cite{w} do not guarantee quality of CU and D2D  communications. Because of interference the communications of the CUs and the D2D users can get severely jeopardized.
\\ 
\indent In all these works, only \cite{k}, \cite{l} and \cite{p}  guarantee QoS requirements for the CUs. Besides, the rate requirements for the D2D users are guaranteed in only \cite{p}. 

\subsection{Contributions}
\indent In this paper, we address the challenging problem of designing a Pareto optimal D2D resource allocation algorithm without inter-D2D interference when DCSI-R is unavailable at the BS. We employ a game theoretic framework for D2D resource allocation and propose two suboptimal semi-distributed algorithms.  Our aim is to maximize the social utility of all the D2D players while ensuring that the QoS requirements of the CUs and the D2D players are satisfied unlike \cite{v} and  \cite{w} which do not guarantee QoS requirements for the CUs and the D2D players. The D2D users who are the players take their actions in a distributed manner and transmit them to the BS, based on which the BS decides which resources are to be allocated to them that have been already allocated to the CUs.  Since part of the D2D resource allocation algorithm runs at the BS while the other part runs at the D2D transmitter, the computational load at the BS reduces. Since our algorithms are semi-distributed DCSI-R doesn't need to be transmitted to the BS which reduces the control overhead and power wastage of the network.\\ 
\indent Our work differs from \cite{t} in that in \cite{t}, multiple players are associated to the same BS. To ensure time fairness the throughput observed by a player is its transmission rate divided by the number of players associated to the BS. As the number of players increase, their throughputs decrease. However, in our algorithm multiple players are not allocated the same RB. The BS allocates orthogonal RBs to the D2D players. Thus, as the players increase, their throughputs do not decrease. If orthogonal RBs are not allocated to the D2D players, inter-D2D interference would occur as in \cite{v} and  \cite{w} which would decrease a D2D player's throughput because of interference from the other D2D players. \\
\indent We demonstrate how interdependence can be designed within the framework of D2D resource allocation.  Unlike  \cite{t}, we design interdependence in a very different manner. In \cite{t}, when multiple players are associated to a BS, a player's throughput decreases  due to the other players because time fairness is ensured by the algorithm to the players which results in interdependence.  In contrast, in our algorithms interdependence is designed by a frame structure in which the D2D players are arranged in a sequence that shifts in every subframe of the frame and by the utilities of the D2D players which are calculated at the end of the frame.  In  \cite{v} and  \cite{w}, since multiple D2D players are allocated the same RB,  inter-D2D interference is caused which brings about interdependence among the D2D players as each D2D player’s throughput gets affected due to the interference from other players. However, this decreases each D2D player’s throughput as the number of players keep increasing. In our algorithms,  this does not happen because interdependence is designed without inter-D2D interference.
\noindent The main contributions of our work are as follows:
 \begin{itemize}
 \item  We apply the distributed algorithm proposed in \cite{t} and develop a semi-distributed suboptimal D2D resource allocation algorithm.  For the channel, we consider a path loss propagation model.
\item In \cite{t}, no randomness in the system model has been considered, which is a challenging problem. We consider shadowing, fast fading and mobility of CUs.  We modify our first proposed suboptimal algorithm to obtain another suboptimal algorithm which gives a robust performance. 
\end{itemize}

\noindent The organization of this paper is as follows. In Section II, we present the system model which comprises of an underlay D2D network and the game theoretic model of the D2D resource allocation problem. In Section III, we propose a suboptimal resource allocation algorithm for the D2D players. In Section IV, we present an example to demonstrate Pareto optimality. Next, in Section V we consider shadowing, fast fading and CU mobility and propose another suboptimal algorithm. We verify the performance of our proposed algorithms through simulations in Section VI. In Section VII, we summarize our findings.
 	\begin{figure}[t]
 	\centering
   	\scalebox{0.8}{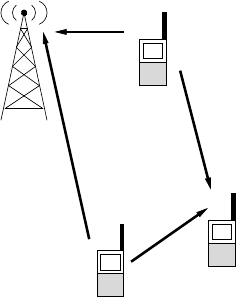}
 	\caption{Resource reuse between a CU and a D2D pair.}
 	\label{fig:1}
 	\end{figure}
 \section{System Model}
\subsection{Network Model}
We consider a single cell scenario with a BS and $N_C$ CUs in uplink transmission mode. We assume that the BS has already allocated RBs to the CUs. Consider $N_D$ D2D pairs which have to be allocated the RBs of the CUs. We assume that a CU is allocated one RB in each subframe and a D2D user can be allocated the RB of only one CU. We denote the BS by $B$, a CU by $c$ and a D2D pair by $d$.
As shown in Fig. 1 \cite{x}, the channel gain between a CU $c$ to the BS $B$ is denoted by $g_{cB}$ and the channel gain between a D2D transmitter to the BS $B$ is given by $g_{dB}$. We assume that both these channel gains or channel state information (CSI) are known at the BS. The channel gain between a D2D transmitter and its receiver is given by $g_{d}$. The channel gain between a CU $c$ and a D2D receiver is given by $g_{cd}$. We assume that the BS does not know the D2D receiver CSI (DCSI-R), i.e. the channel gains $g_{d}$ and $g_{cd}$ but the D2D receiver knows them through some technique. We consider the channel to exhibit path loss.\\
\indent Let a CU's transmit power and a D2D user's transmit power be $P_C$ and $P_D$ respectively.
Given that a D2D user $d$ is allocated the RB of a CU $c$, the signal-to-interference-plus-noise ratio (SINR) of the CU $c$ at the BS is expressed as,
\begin{eqnarray}
\gamma_{c}=\frac{P_C \, g_{cB}}{P_{D} \, g_{dB}+N_0},
\end{eqnarray}
where $N_0$ is the average noise power and $P_{D} \, g_{dB}$ is the interference power received from a D2D user.

\subsection{Game Theoretic Model}
\indent Let $\mathcal{G}$ be a strategic form game with $N_D$ D2D players constituting a set $\mathcal{N}_D=\{d_1, \, d_2, \, ... \, , \, d_{N_D}\}$. Each D2D player $d$ can choose an action from a set of finite actions $\mathcal{A}_d$. Let the joint action set be $\mathcal{A}=\mathcal{A}_d^{N_D}$ and the utility function be $U_d: \mathcal{A}\rightarrow \mathbb{R}$. \\
\indent As per the LTE standard, time is divided into intervals of 1 ms duration called transmission time intervals (TTIs) or subframes. Let the subframe be indexed by $\tilde{n}$. We define a frame to be a time window consisting of $N_D$ subframes.
We number a frame by $n$. Both $\tilde{n}$ and $n$ start from 1. At the start of a frame $n$, the action of every D2D player $d$ is to select a list $l_d(n)$ and transmit it to the BS.  The list contains $K$ possible CUs drawn from a set $\mathcal{N_C}=\{c_1, \, c_2, \,...\,, \, c_{N_C}\}$ of available CUs to form a $K$-tuple. Hence, the total number of lists that a D2D player can generate is $K$ permutations of $N_C$ CUs,  $L=\perm{N_C}{K}$ which constitutes the action set $\mathcal{A}_d$ of a D2D player. An action profile is a tuple of the actions taken by the D2D players.
 Let an action profile be denoted as $l=(l_1,\, l_2,\, ... \, ,\, l_{N_D}) \in \mathcal{A}$. We denote
$l_{-d}$ to be the action profile of the D2D players other than the D2D player $d$. Therefore, the action profile can also be written as $l=(l_d,\, l_{-d})$. \\
\indent We design an algorithm at the BS with the help of which the BS decides which CU's RB is to be allocated to a D2D player in each subframe $\tilde{n}$ as per its list. Using this CU's RB, a D2D transmitter then starts transmitting. A D2D player $d$ then observes a throughput of $R_d(\tilde{n})$ given by,
\begin{eqnarray}
R_d(\tilde{n})=B\log_2\bigg(1+\frac{P_D \, g_{d}}{P_{C} \, g_{cd}+N_0}\bigg)
\end{eqnarray}
Utility of a D2D player is the payoff that it receives at the end of a frame and we define it as follows.
\begin{definition}
\emph{(Utility)} At the end of a frame $n$, the average of the throughputs obtained by a D2D player $d$ over the subframes of a frame  is calculated and is given by,
\begin{eqnarray}
r'_d(n)=\frac{1}{N_D}\sum_{\tilde{n} = uN_D+1 }^{(u+1)N_D}R_d(\tilde{n}),
\end{eqnarray}
where $uN_D+1$ denotes the start of a frame, $(u+1)N_D$ denotes the end of a frame and $u=0,1,2, \ldots$ such that $u=0$ stands for frame 1, $u=1$ for frame 2 and so on for the subsequent frames. This is normalized by a normalization factor $\alpha$ to get the utility $r_d(n)$,
\begin{eqnarray}
r_d(n)=\frac{1}{\alpha}r'_d(n),
\end{eqnarray}
such that $0 \le r_d(n) < 1$, \cite{e}.
\end{definition}
\noindent The algorithm at the BS ensures that the utilities and the actions of the D2D players are aligned, that is, the utility of each D2D player is a function of the lists of all the D2D players (action profile $l$), $r_d(n)=U_d(l)$.
 
\begin{definition}
\emph{(Social Utility)} The social utility is the sum of the utilities of the D2D players obtained in a frame $n$ and is given by $W_{l}(n)= \sum_{d\in\mathcal{N_D}}r_d(n)$.
\end{definition}

\noindent Our objective is to maximize the social utility of the D2D players in every frame $n$, over all possible action profiles, to get the optimal action profile $l^*(n)$, subject to a minimum utility constraint of $r_d^{tgt}$ for each D2D player and a minimum SINR constraint of $\gamma^{tgt}$ for each CU in every subframe,
\begin{align}
& \quad l^*(n)=\argmax_{l \in \mathcal{A}} \,  \sum_{d\in\mathcal{N_D}}r_d(n), \\
\text{s.t.} & \quad r_d(n) \geq  r_d^{tgt}, \quad \forall  \, \, d \in \mathcal{N_D}, \nonumber\\
 & \quad \gamma_{c}(\tilde{n}) \geq  \gamma^{tgt}, \quad \forall  \, \, c \in \mathcal{N_C}. \nonumber
\end{align}
Thus, the optimal action profile $l^*(n)$ is the set of Pareto optimal actions taken by the D2D players such that the social utility of the D2D players is maximized and the utility constraint for each D2D player and the SINR constraint for each CU is satisfied.\\
\indent To understand Pareto optimality, refer to the example of Prisonner's dilemma game in  \cite{e}. The sum of utiilities or payoffs of the prisoners which is termed as the social utility is maximum in this game when the actions of the prisoners are to confess to the crime. Thus, their actions to confess are Pareto optimal or welfare maximizing as compared to the Nash Equilibrium which does not maximize their welfare. As per \cite{e} and  \cite{t},  their proposed algorithms converge to a Pareto optimal solution only when interdependence is ensured in the game structure. Since our proposed algorithm has been developed from the algorithm of \cite{t}, this is applicable to our proposed algorithm as well.

\begin{definition}
\emph{(Interdependence, \cite{e})}
A game $\mathcal{G}$ is interdependent if for every action profile $l \in \mathcal{A}$ and for every proper subset of players $\mathcal{N} \subset \mathcal{N}_D$, there exists a player $d\notin \mathcal{N}$ and a choice of actions $l'_{\mathcal{N}} \in \mathcal{A}_d^{|\mathcal{N}|}$ such that $r_d(l'_{\mathcal{N}},l_{-\mathcal{N}})\neq r_d(l_{\mathcal{N}},l_{-\mathcal{N}})$.
\end{definition}

\section{Resource Allocation Algorithm}
The resource allocation algorithm that we propose is a game theoretic learning algorithm based on the algorithm proposed in \cite{t}.
We consider the internal state variables of a D2D player to be its list $l_d(n)$, its utility $r_d(n)$ and another variable called its mood $m_d(n)$. The mood can take two values: content ($C$) and discontent ($D$). The collection of states of all the D2D players constitutes the system state $s(n)=(s_1(n), \, ... \,,\, s_{N_D}(n))$. Let the action profile in every frame $n$ be $l(n) = (l_1(n), \, ... \, ,\, l_{N _D}(n))$, the utility profile be $r(n) = (r_1(n), \, ... \, ,\, r_{N_D}(n))$ and the mood profile be $m(n) = (m_1(n), \, ... \, ,\, m_{N_D}(n))$. A stochastically stable state is a tuple of the optimal action profile (the optimal lists), the utility profile which satisfies the utility constraints of of the D2D players and the SINR constraints of the CUs, where the sum of utilities of the D2D players is the maximum D2D social utility and the mood profile in which the moods of the D2D players are content. This solves the problem stated in $(5)$.
\begin{algorithm}[t]
	 \caption{Resource Allocation Algorithm}\label{alg:1}
	 \begin{algorithmic}[1]
		\ForAll{$\tilde{n} =1,2,\dots$}
		\If{$\tilde{n}\: mod \: N_D==1$}
		\Procedure{ListSelection}{$m_d(n-1),l_d(n-1)$}
		 \If{$m_d(n-1)==C$}
		 \State $l_d(n) \gets i$, \quad w.p. $\frac{\epsilon^k}{L-1}$, \quad $\forall \,\, i \in \mathcal{A}_d \backslash l_d(n-1)$
		 \State $l_d(n)\gets l_d(n-1)$, \quad w.p. $1-\epsilon^k$
		 \Else
		 \State $l_d(n) \gets i$, \quad w.p. $\frac{1}{L}$,\quad $\forall \, i \in \mathcal{A}_d$
		 \EndIf
		 \EndProcedure
		 \State D2D player $d$ transmits its list $l_d(n)$ to the BS.
		 \State $h\gets1$
		 \EndIf
		 \Procedure{AllocationBS}{$l(n),D2DP \textunderscore seq$}
		 \State $\bm{a}\gets$ \textsc{AllocBS}$(l(n),D2DP \textunderscore seq)$
		 \State Left circular shift $D2DP\textunderscore seq$ by 1
		 \State $c\gets \bm{a}(d)$ \Comment{BS transmits to the  D2D player $d$  that it  is allocated CU $c$'s RB.}
		 \EndProcedure
		 \State D2D player $d$ transmits using CU $c$'s RB, observes throughput $R_d(\tilde{n})$ and normalizes it to $r'_d(\tilde{n})$.
	 	 \State  $\bm{r}_d(h)\gets r'_d(\tilde{n})$ \Comment{D2D player $d$ stores $r'_d(\tilde{n})$ obtained in each subframe $		\tilde{n}$ in vector $\bm{r}_d$. }
		  \State $h\gets h+1$
		  \If{$\tilde{n}\: mod \: N_D==0$}
		  \Procedure{UtilityCalculation}{$\bm{r}_d$}
		  \State $r_d(n)=sum(\bm{r}_d)/N_D$
		  \EndProcedure
		  \Procedure{MoodCalculation}{$s_d(n-1), l_d(n)$, $r_d(n)$}
		  \If{$m_d(n-1)==C \ \& \ [l_d(n-1),r_d(n-1)]==[l_d(n),r_d(n)]$}
		  \State $m_d(n)\gets C$
		  \ElsIf{$r_d(n)\ge r_d^{tgt}$}
		  \State $m_d(n)\gets C$,  w.p. $\epsilon^{1-r_d(n)}$
		  \State $m_d(n)\gets D$,  w.p. $1-\epsilon^{1-r_d(n)}$
		  \Else
		  \State $m_d(n)\gets D$
		  \EndIf
		  \EndProcedure
		  \EndIf
		  \EndFor
	 \end{algorithmic}
	 \end{algorithm}
	 
\begin{theorem}
The stochastically stable states of a regular perturbed DTMC are the states $s^* \in \mathcal{S}$, which satisfy the following conditions: \\
\emph{(1)}~The action profile $l(n)$ maximizes the social utility $W_{l}=\sum_{d\in\mathcal{N_D}}r_d(l)$ while satisfying the utility constraints of the D2D players and the SINR constraints of the CUs.\\
\emph{(2)}~The utility $r_d(n)$ of a D2D player $d$ is aligned to the action profile $l(n)$, that is, $r_d(n)=U_d(l)$.\\
\emph{(3)}~The mood $m_d(n)$ of each D2D player $d$ is content.
\end{theorem}

\begin{proof}
The proof is in the same lines as \cite{t}, provided interdependence is ensured in the game.
\end{proof}
The resource allocation algorithm, Algorithm 1, which we will discuss next consists of the following stages: 1)~the list selection of a D2D player, 2)~the BS's resource allocation algorithm, 3)~the utility calculation of a D2D player and 4)~it's mood calculation method.

	\begin{figure}[t]
 	\centering
   	\scalebox{0.4}{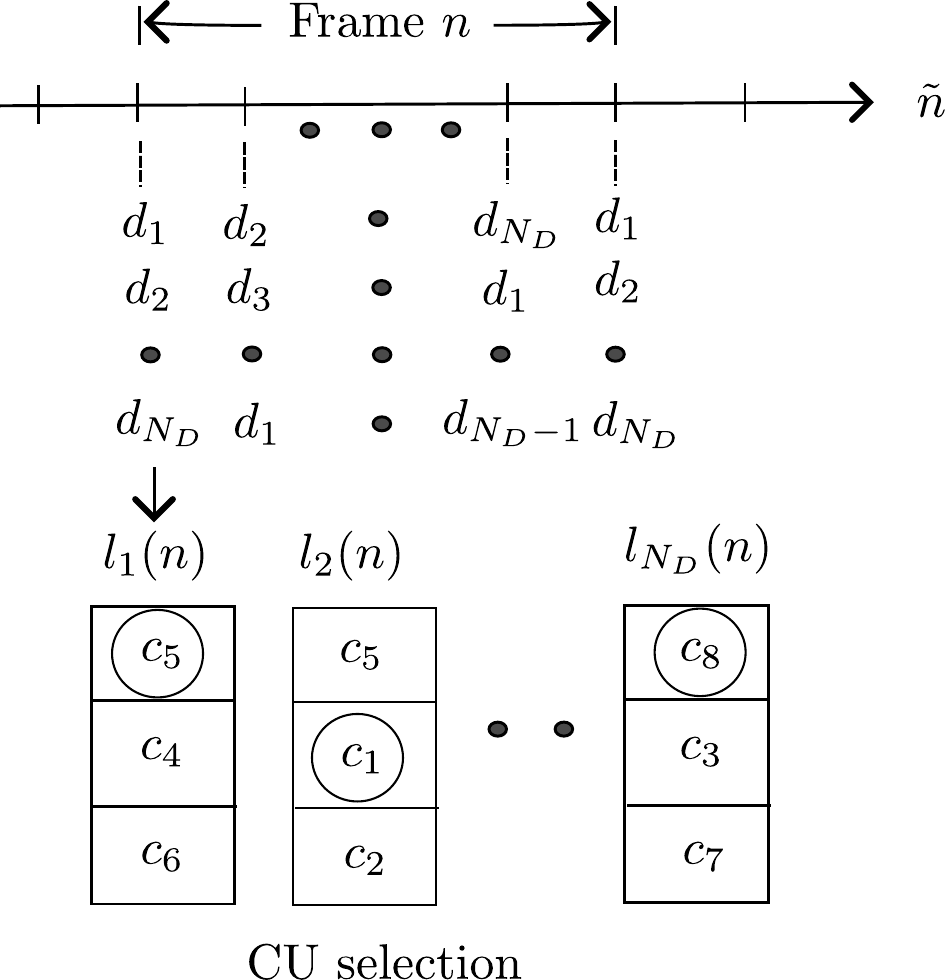}
 	\caption{D2D players' sequencing and the selection of CUs from the lists of the D2D players in the subframes of a frame.}
 	\end{figure}
 			  
\subsection{List Selection}
Every D2D player selects a list at the start of every frame. The mood of a D2D player in the previous frame helps it to select its list in the present frame. Let $\epsilon$ be the experimentation rate, such that $\epsilon > 0$ and $k$ is a constant such that $k>N_D$\cite{e}. A D2D player whose mood is content can either decide to retain (exploit) its previous list with probability (w.p.) $1-\epsilon^{k}$ or select (explore) a new list w.p. $\epsilon^k/(L-1)$.  It is more likely to retain its previous list rather than exploring other lists if $\epsilon^k$ is less than $1-1/L$. Increasing $\epsilon$ results in an increase in exploration. 
However, increasing $k$ results in a decrease in $\epsilon^k$ due to which the probability of exploring becomes less. Thus, $\epsilon$ and $k$ are important design parameters that decide on the frequency of exploring versus exploiting lists. A D2D player whose mood is discontent explores all the $L$ possible lists and chooses a random list w.p. $1/L$. Steps 3 - 10 of Algorithm 1 demonstrate the list selection method for a D2D player $d$. Once this is over for all the D2D players, every D2D player transmits its list to the BS. The BS then allocates the RBs of the CUs to the D2D players as per the following algorithm (refer Algorithm 2).
	 
\subsection{Resource Allocation Algorithm at the BS}
\noindent \emph{1) D2D players' sequencing in the subframes of a frame:} At the beginning of every frame, the lists of all the D2D players are available at the BS. The BS follows a sequence of D2D players which we term as the $D2DP \_seq$ in every subframe of a frame so that it can select their lists as per this sequence. The $D2DP \_seq$ in every subframe shifts left by one (left circular shift) as compared to the $D2DP \_seq$ in the previous subframe. This is done to prioritize the lists of the D2D players in every subframe. In the first subframe of every frame the BS orders the D2D players in a sequence, $D2DP \_ seq=d_1,d_2, \, ... \,, \, d_{N_D}$ as shown in Fig. 2. In the next subframe, the $D2DP \_seq$ starts from $d_{2}$ such that $D2DP \_ seq=d_2,\, ... \,, \,d_{N_D},d_1$ and so on till the last subframe of a frame when the $D2DP \_seq$ starts from $d_{N_D}$. In this way, the selection of the first D2D  player in the $D2DP \_seq$ changes in every subframe of a frame.

\begin{algorithm}[t]
	 \caption{Resource Allocation Function at the BS}\label{alg:2}
	 \begin{algorithmic}[1]
		\Statex \textbf{Input:  }$l(n), D2DP  \textunderscore seq$
		\Statex \textbf{Output:  }$\bm{a}$
		\Function{AllocBS}{$l(n), D2DP \textunderscore seq$}
		\ForAll{$j =1\: to \: N_D$}
		\State $d\gets D2DP \textunderscore seq[j]$
		\State $\bm{a}(d)\gets0$ \Comment{No CU is mapped to a D2D player  $d$ initially.}
		\State BS selects D2D player $d$'s list $l_d(n)$.
		\State $x\gets1$ \Comment{$x$ is the index of the D2D player $d$'s list.}
		\While{$x \neq list \textunderscore length+1$}
		\State  BS selects element $c$ from the $x^{th}$ index of $l_d(n)$.
		\If{CU $c$'s RB is not assigned to the D2D players  before D2D player $d$ in $D2DP \textunderscore seq$}
	 	\State $\bm{a}(d)\gets c$  
	 	\State \textbf{break} 
		\Else		
		\State $x\gets x+1$
		\EndIf 	
		\EndWhile
		\If {$\bm{a}(d) \neq 0$}			
		\If{$\gamma_c<\gamma^{tgt}$}\Comment{Allocation Test}
		\State $\bm{a}(d)\gets0$ \Comment{No CU is mapped to a D2D player  $d$.}
		\EndIf
		\EndIf
		\EndFor
		\State \Return $\bm{a}$
		\EndFunction
 	\end{algorithmic}
 \end{algorithm}
 
\noindent \emph{2) Orthogonal allocation of RBs:}
In every subframe, the BS ensures orthogonal allocation of CUs' RBs. In a subframe, it selects the first element (a CU) from a D2D player's list  and allocates this CU's RB to it if it is not allocated to players prior to it in the $D2DP\_seq$ followed in that subframe (see Fig. 2). If this element has already been assigned to another D2D player, the BS selects the next element from the D2D player's list. This goes on till it reaches the end of the list. In case all the elements of its list are already assigned to the other players, it is not allocated any RB (refer Algorithm 2).\\
\indent   From Fig. 2, consider the first subframe and the D2D player $d_1$ which gets the first priority. We must note that  in the first subframe, the first D2D player $d_1$'s allocation of CU $c_5$'s RB  in $D2DP \_seq$  is not getting affected due to the allocation of  CUs' RBs  to the other D2D players. But in the second subframe its priority becomes the least ($N_D$), that is, it is the last element of the  $D2DP \_seq$ as the BS's priority shifts to the next D2D player $d_2$. Thus, its allocation of CU's RB gets affected due to the allocation of CUs' RBs to the other D2D players who are before it in the $D2DP \_seq$.  In this way, in a frame it is ensured that the allocation of a CU's RB to a D2D player and thus its throughput gets affected by the allocation of CUs' RBs to the other D2D players. This is done to ensure interdependence in the game. We must note that the allocation of CUs' RBs depends on the D2D players' lists. Thus, a D2D player's throughput depends on the lists or actions of other D2D players.\\
\indent We must also note the importance of list selection by a D2D player. If a D2D player does not send a list and instead selects a CU and conveys it to the BS, then the BS in order to ensure orthogonal allocation of RBs, can communicate to it that this CU's RB has already been allocated to another D2D player and so it needs to choose another CU. The D2D player again has to convey this to the BS. This can happen many times, which increases the communication overhead between the BS and the D2D player.
Instead of this, if the D2D player sends a list to the BS, then the BS can check if the first element of its list is already allocated to another D2D player. If so, then it can select the next element of the D2D player's list.
This is equivalent to communicating another choice of CU by the D2D player to the BS. Hence, sending a list reduces the control overhead and conserves the BS's and the D2D player's power and resources.
 			
\noindent \emph{3) Allocation Test:}
When the BS assigns a CU's RB to a D2D player, in order to ensure that the CU's communications are not hampered in each subframe the BS checks whether a CU's SINR will decrease below its SINR target $\gamma^{tgt}$ due to interference from the D2D player which has been assigned this CU's RB. If the SINR of the CU $\gamma_{c}$ decreases below $\gamma^{tgt}$, the BS does not allocate the CU's RB to the D2D player. Consequently, the D2D player's throughput $R_d(\tilde{n})$ is zero in that subframe. After this test, the BS conveys to each D2D player which CU's RB is allocated to it  (refer Algorithm 2). The mapping of CUs to the D2D players in each subframe is termed as an allocation profile and is denoted by $a$. For example, from Fig. 2 if we assume that in the first subframe  the CUs $c_5, \, c_1, \, ... \,,c_8$ selected by the BS from the lists $l_1(n)$ to $l_{N_D}(n)$ (action profile) of the D2D players pass the allocation test, then the allocation profile in the first subframe is $a=(c_5, \,c_1,  \, ... \,,c_8)$. Since an allocation profile is generated in every subframe of a frame and there are $N_D$ subframes in a frame, a set of $N_D$ allocation profiles are generated in a frame from the action profile.

\subsection{Utility Calculation}
In a subframe $\tilde{n}$, a D2D player transmits using the RB of the CU allocated to it by the BS and observes a certain throughput $R_d(\tilde{n})$ which it normalizes to $r'_d(\tilde{n})$. At the end of a frame, it calculates an average of its normalized throughputs over the subframes which is its utility $r_d(n)$.\\
\indent When a D2D player's utility is calculated at the end of a frame as an average of its throughputs, since its throughput gets affected due to the lists (actions) of others, its utility also gets affected due to the lists of others. As per Definition 3, this ensures interdependence because a player's utility gets affected by the actions of others.
 Thus, interdependence is enforced among the players because of the orthogonal allocation of CUs' RBs in a subframe as per the $D2DP\_{seq}$, the prioritization technique of D2D players across subframes and the way in which the utility is calculated at the end of every frame. \\
\indent However, it may happen that the first elements of all the lists of the D2D players are different. In this case, there is no interdependence in the game. This can occur over some frames. When the lists change again and this case does not occur, again interdependence is ensured in the game. Therefore, our proposed algorithm is suboptimal. 

\subsection{Mood Calculation}
We define the present state $\bm{s}_d(n)$ of a D2D player $d$ as a 3-D tuple, $\bm{s}_d(n)=(l_d(n), \, r_d(n), \,m_d(n))$. The mood of a D2D player is determined (refer lines 26 - 35 of Algorithm 1) at the end of each frame from its previous state, its present list and utility. A content D2D player will remain content, if its present configuration does not not change with respect to its previous one, that is, $[l_d(n), \, r_d(n)]=[l_d(n-1), \, r_d(n-1)]$. However this condition is violated if $l_d(n)\neq l_d(n-1)$ or $r_d(n) \neq r_d(n-1)$ or both. We next explain these cases. Suppose in the present frame $n$,  a content D2D player $d$ decides to explore other lists. Then, its list $l_d(n)$ changes and hence $l_d(n)\neq l_d(n-1)$. Thus, its present configuration changes. If this player decides to retain its previous list, $l_d(n)= l_d(n-1)$, its present utility $r_d(n)$ can still change with respect to its previous utility $r_d(n-1)$ because if some other player changes its list, then it can change the utility of this player because of interdependence and hence $r_d(n) \neq r_d(n-1)$. Therefore, in these cases as its present configuration changes, the chances of becoming discontent or content depends on its present utility $r_d(n)$. If $r_d(n)$ is greater than the minimum required utility $r_d^{tgt}$, it remains content w.p.  $\epsilon^{1-r_d(n)}$ or it becomes discontent  w.p.  $1-\epsilon^{1-r_d(n)}$. The chances of remaining content are more than that of becoming discontent, if $r_d(n)$ is sufficiently high. If $r_d(n)$ is less than $r_d^{tgt}$, then it becomes discontent w.p. one. On the other hand, if a discontent D2D player $d$'s utility $r_d(n)$ is greater than $r_d^{tgt}$, it becomes content  w.p. $\epsilon^{1-r_d(n)}$ or it remains discontent w.p.  $1-\epsilon^{1-r_d(n)}$. If $r_d(n)$ is less than $r_d^{tgt}$, then it remains discontent w.p. one.
 
\section{Concept Illustration}
We demonstrate through the following example how the actions of D2D players can converge to a Pareto optimal solution which maximizes their social utility. The action of a D2D player is to choose a CU and its utility is its throughput. We consider the D2D players' social utility to be their sum throughput. In order to theoretically calculate the optimal action profile we assume that all the channel gains are known so that we can determine the throughputs of the D2D players. We can then determine their Pareto optimal actions which would maximize their social utility.  We show through simulations how our proposed algorithm performs.\\
\noindent \emph{1) Topology:}
We consider a simple topology with three CUs $c_1$, $c_2$ and $c_3$ positioned in a semi-circle of radius $R_1$ as shown in Fig. 3 and three D2D pairs $d_1$, $d_2$ and $d_3$ in an outer semi-circle of radius $R_2$.\\
\emph{2) Parameter Setting:}
Let $R_1=50$ m and $R_2=100$ m. Every D2D transmitter subtends an angle of 10 degrees at the BS with respect to its receiver. The distance between a D2D transmitter and its receiver can be calculated from Fig. 3 and is $17.43$ m. Let $P_D$ and $P_C$ be 10 mW. The path loss model be $1/D^\eta$ where $D$ is the distance between a transmitter and a receiver and $\eta$ is the path loss exponent which we set to $2$. Let $N_0=0.1$ mW, bandwidth $B=1$ Hz and $\gamma^{tgt}= 0$ dB. We set $r_d^{tgt}$ of each D2D player to $50$ $\%$ of the rates that they achieve when the action profile is $(c_1,c_2,c_3)$. Let the length of the lists be $K=3$.
	\begin{figure}[t]
 	\centering
 	\scalebox{0.4}{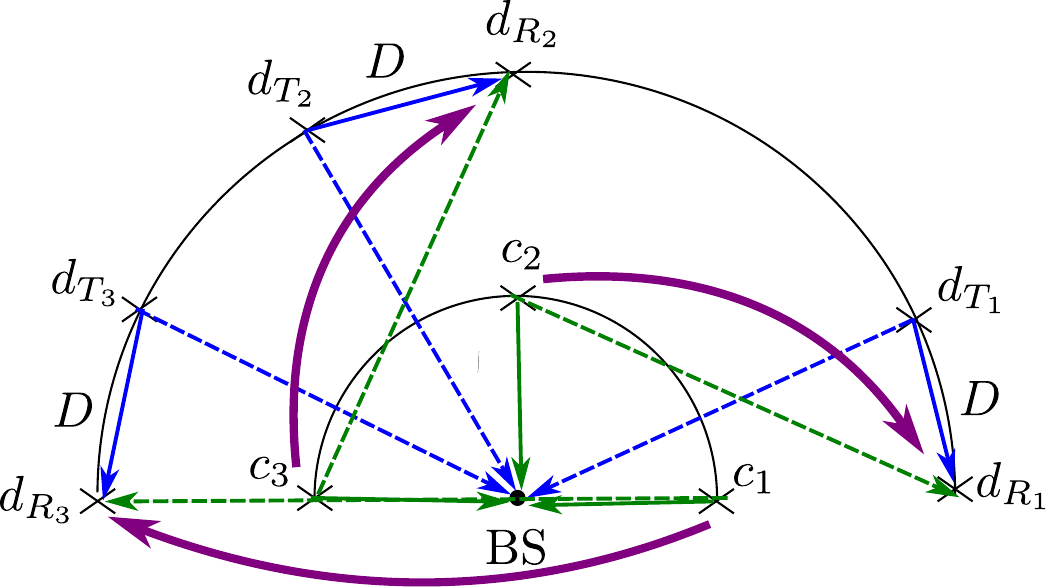}
 	\caption{Example demonstrating the optimal mapping of CUs and the D2D pairs.}
 	\end{figure}
 	\\
\emph{3) Analysis:} Since the distances of the three CUs to the BS are $R_1$, the powers received from them at the BS are the same. The interference powers received from the D2D transmitters at the BS are also the same because they are also equidistant from the BS. Hence, the SINRs of all the CUs are the same no matter what the action profile is. Since these SINRs are higher than $\gamma^{tgt}$, the minimum SINR constraint of the CUs is met. Therefore, the D2D players do not hamper CUs' communications. The best CU for a D2D player is the one which is farthest from a D2D player's receiver. If CU $c_1$'s RB is allocated to D2D pair $d_3$, then CU $c_2$'s RB can be allocated to either $d_1$ or $d_2$. If $c_2$'s RB is allocated to $d_2$, then $d_2$ faces the worst case interference from $c_2$ while $d_3$ gets $c_3$'s RB which is the best CU for it. This allocation is not Pareto optimal because while $d_1$ and $d_3$ get their best choices of CUs, $d_2$ gets the worst. Thus, $c_2$'s RB should be allocated to $d_1$ and $c_3$'s RB should be allocated to $d_2$. This makes the allocation Pareto optimal and therefore the action profile $a_1^*=(c_2,c_3,c_1)$ is Pareto optimal. The action profile $a_2^*=(c_3, c_1,c_2)$ is also Pareto optimal. With action profile $a_1^*$, the throughputs of the D2D players are $r_{d_1}= 0.4076$, $r_{d_2}=0.4076$ and $r_{d_3}=0.4089$ bps. Thus, the optimal utility profile $r(n)^*$ is $(0.4076, 0.4076, 0.4089)$. Therefore, the maximum value of the social utility is the sum of these throughputs which is $1.2241$ bps. \\
\emph{4) Simulation:}
We simulate Algorithm 1 with the above topology and parameter setting. We set $\epsilon=0.5$ and $k=11$.  Let the mood of a D2D player $m_d=1$ when it is content and $m_d=0$ when it is discontent. Normalized mood is the sum f the moods of all the players divided by a factor which we have set to $15$. We then obtain a plot of the social utility (left y axis) and the normalized mood (right y axis) of the D2D players across the frames as shown in Fig. 4. We observe from Fig. 4 that the maximum value of the social utility of the D2D players is $1.224$ bps. However, since our algorithm is suboptimal the number of frames over which the social utility is maximum is less. If the algorithm had been optimal,  the system would have spent most of its time in its optimal states \cite{t}.
	\begin{figure}[t]
 	\centering
 	\scalebox{0.6}{\includegraphics{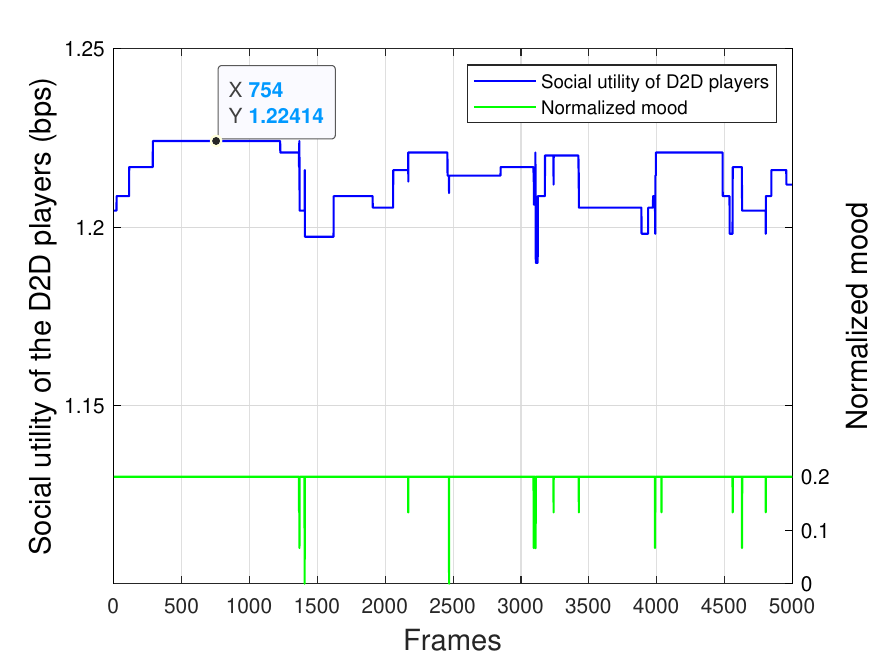}}
 	\caption{Social utility of the D2D players over frames for the concept illustration example.}\label{fig:2}
	\end{figure}

\section{Shadowing, Fast Fading and Mobility of CUs}
We now consider shadowing, fast fading and mobility of CUs. When the channel is a fast fading channel and the CUs are moving, it is  a challenge to devise efficient resource allocation algorithms  for the D2D users. The mobility of a CU is characterized by its velocity and direction of movement. When a CU is moving, the distance of a CU to a D2D player's receiver changes with time.  Because of this and a fast fading channel, a D2D player's SINR and consequently its throughput varies over subframes when it is allocated a CU's RB. Algorithm 1 when applied to this case, fails to maintain any state over a certain period of time and fluctuates continuously. We therefore modify Algorithm 1  to propose another suboptimal algorithm, Algorithm 3.  We devise a novel utility calculation method which is threshold based. The main concept is to detect a significant change in the network topology and the channel due to which the resource allocation to the D2D players should change. 

\begin{algorithm}[t]
	 \caption{Threshold Based Resource Allocation Algorithm}\label{alg:2}
	 \begin{algorithmic}[1]
		 \State Initialize: $\bm{n}_d \gets [0 \ldots 0]$, $\bm{u}_d \gets [0 \ldots 0]$, $\bm{r}_a\gets [0 \ldots 0]$,  $k\gets1$
		 \ForAll{$\tilde{n} =1,2,\ldots$}
		 \If{$\tilde{n}\: mod \: (10N_D)==1$}
		 \State D2D player $d$ selects a list $l_d(\hat{n})$ at the start of a superframe $\hat{n}$.
		 \EndIf
		 \State $\bm{a}\gets$ \textsc{AllocBS}$(l(\hat{n}),D2DP\textunderscore seq)$ \Comment{\textsc{AllocBS()}: Resource allocation function at the BS.}
		 \State Left shift $D2DP\textunderscore seq$ by 1   
		 \State $c\gets\bm{a}(d)$ \Comment{BS transmits that D2D player $d$ is allocated CU $c$'s RB.}
		 \State D2D player $d$ transmits using CU $c$'s RB, observes throughput $R_d(\tilde{n})$ and normalizes it to $r'_d(\tilde{n})$.
		 \State $\bm{a}_d(k) \gets c$
		 \If{$\bm{a}_d(k) \neq 0$}
		 \State $\bm{n}_d(c)\gets \bm{n}_d(c)+1$
		 \State $\bm{r}_a(c) \gets \bigg(1-\frac{1}{\bm{n}_d(c)}\bigg)\bm{r}_a(c)+\frac{1}{\bm{n}_d(c)}r'_d(\tilde{n})$
		 \EndIf 
		 \State $k\gets k+1$
		 \If{$k==N_D+1$}
		 \State $k\gets1$
		 \EndIf
		 \If{$\tilde{n}\: mod \: (10N_D)==0$}
		  \ForAll{$c=1 \: to \: N_C$}
		  \If{$\bm{u}_d(c)==0$}       
		  \State $\bm{u}_d(c) \gets \bm{r}_a(c)$       
		  \Else       
		  \If{$abs(\bm{u}_d(c)-\bm{r}_a(c))>\Delta$ }
		  \State $\bm{u}_d(c) \gets \bm{r}_a(c)$      
		  \EndIf 	   	
		  \EndIf
		  \EndFor
		  \State Take the non-zero elements of $\bm{a}_d$ and put it in a vector  $\bm{a}'_d$.
		  \State $r_d(\hat{n}) \gets sum(\bm{u}_d(\bm{a}'_d))/N_D$
		  \State D2D player  $d$ calculates its mood $m_d(\hat{n})$.
		  \EndIf
		  \EndFor
 	\end{algorithmic}
 \end{algorithm}

\subsection{Threshold Based Resource Allocation Algorithm}
\indent We modify Algorithm 1 by repeating a frame ten times to form a superframe which is numbered by $\hat{n}$. Index $\hat{n}$ starts from one. The list selection by a D2D player occurs at the beginning of each superframe, the resource allocation algorithm at the BS remains the same, the utility calculation method which will now differ and the mood calculation occur at the end of a superframe. In Algorithms 1 and 2, the index $n$ for a frame will be replaced with $\hat{n}$ for a superframe. Next,  we explain the utility calculation method (refer Algorithm 3). \\
\indent In any subframe $\tilde{n}$, two pieces of information are available to a D2D player $d$: 1) the $c^{th}$ CU whose RB is allocated to it and 2) the D2D player $d$'s  throughput $R_d(\tilde{n})$ when it transmits using this CU's RB which the D2D player $d$ normalizes. In a continuum of subframes, a CU $c$'s RB can be allocated to the D2D player $d$ several times. Every time the D2D player $d$ is allocated this CU's RB, it observes a change in its throughput as a result of the CU's mobility and the channel. In the first frame of length $N_D$ of a superframe, a D2D player is allocated the RBs of different CUs. It stores this sequence of CUs in a vector $\bm{a}_d$ of length $N_D$. This sequence is again repeated in the next frames of the superframe. Every time it observes that a CU $c$'s RB is allocated to it, it calculates a running average of its previous normalized throughputs observed with this CU over subframes and the present normalized throughput using Monte Carlo (MC) averaging and stores it in a vector $\bm{r}_a$ of length $N_C$ at its $c^{th}$ index. It also stores the number of times that it is allocated this CU $c$'s RB in a vector $\bm{n}_d$ of length $N_C$ at its $c^{th}$ index. This is required for MC averaging. This averaging is also done for all the other CUs whose RBs are allocated to it. When this CU $c$'s RB is allocated for the first time to it, then its throughput is stored in $\bm{r}_a$ at its $c^{th}$ index. It also maintains a vector $\bm{u}_d$ of length $N_C$ which we term as the utility vector. At the end of every superframe, it checks the entries of $\bm{u}_d$ and in whichever indices of $\bm{u}_d$ it finds a zero, it stores the contents of $\bm{r}_a$ from these indices to $\bm{u}_d$. If it finds non-zero entries in some indices of $\bm{u}_d$, then it calculates an absolute of the difference of these entries with those of $\bm{r}_a$. If any of the resultants exceed a threshold $\Delta$, then those indices of $\bm{u}_d$ are mapped to the indices of $\bm{r}_a$ and $\bm{u}_d$ is updated with the corresponding values of $\bm{r}_a$ at these indices. Thus, the decision to update the entries of the utility vector $\bm{u}_d$ with those of $\bm{r}_a$ depends on $\Delta$. \\
\indent At the end of a superframe $\hat{n}$, a D2D player retrieves the contents of $\bm{u}_d$ from those indices which are in turn the non-zero elements of $\bm{a}_d$. It then sums up these contents of $\bm{u}_d$ and divides it by $N_D$ to obtain its utility $r_d(\hat{n})$. Its mood  $m_d(\hat{n})$ is then determined  as per the mood calculation method of Algorithm 1 from its utility, the list selected by it in the present superframe and its previous state.

\section{Results}
In this section, we evaluate the performance of our proposed algorithms. We consider a circular cell model whose radius is $250$ m in which the CUs and the D2D transmitters are uniformly distributed.  A D2D receiver is uniformly distributed around a D2D transmitter within a circle of radius $50$ m which is the range of D2D communications. The transmit power of a CU is $P_C=250$ mW and that of a D2D player is $P_D=1$ mW. The bandwidth $B$ is $180$ kHz. We assume that in each subframe a CU is allocated one RB. We consider the LTE path loss model, $PL=128.1+37.6 \, log\ (D)$. The UE and the BS's noise figures are $9$ dB and $5$ dB respectively. The thermal noise density is $-174$ dBm. It should be noted that the total number of lists $L$ from which a D2D player can select a list depends on the the length of the lists $K$ because $L=\perm{N_C}{K}$. If $K$ is large, then $L$ increases and so the number of bits required to transmit a list increases. Therefore, we set $K$ to $3$.

\subsection{With Path Loss}
The number of CUs $N_C$ and the number of D2D pairs $N_D$ are set to $5$. Let $\gamma^{tgt}$ be $10$ dB and the normalization factor $\alpha$ be $10$ Mbps. The minimum required utility $r_d^{tgt}$ for each D2D player is set to $10$ kbps, normalized by the normalization factor $\alpha$. When $N_C$, $N_D$, $\gamma^{tgt}$ and  $r_d^{tgt}$ are increased more than this, we have observed in the simulations that the plot of the social utility of the D2D players over frames fluctuates continuously. 
Fig. 5 demonstrates the performance of Algorithm 1 with path loss. We plot the social utility of the D2D players over frames with three different values of $\epsilon$ and $k$. When choosing these design parameters we must note that the exploration probability $\epsilon^k/(L-1)$ depends on them. If it is very high then the D2D players explore more often due to which the duration for which the system states are maintained becomes very small. If it is very low then the D2D players explore so less that the transition from one system state to another occurs infrequently. We set $\epsilon=0.68$ and $k=21$ and then increase $\epsilon$ to $0.69$ which results in the exploration probability to increase.  Thus, the D2D players explore lists more often. Due to a change in their actions their individual utilities also change more often which is reflected in the fluctuations observed in the social utility of the D2D players.

 	\begin{figure}[t]
 	\centering
 	\scalebox{0.6}{\includegraphics{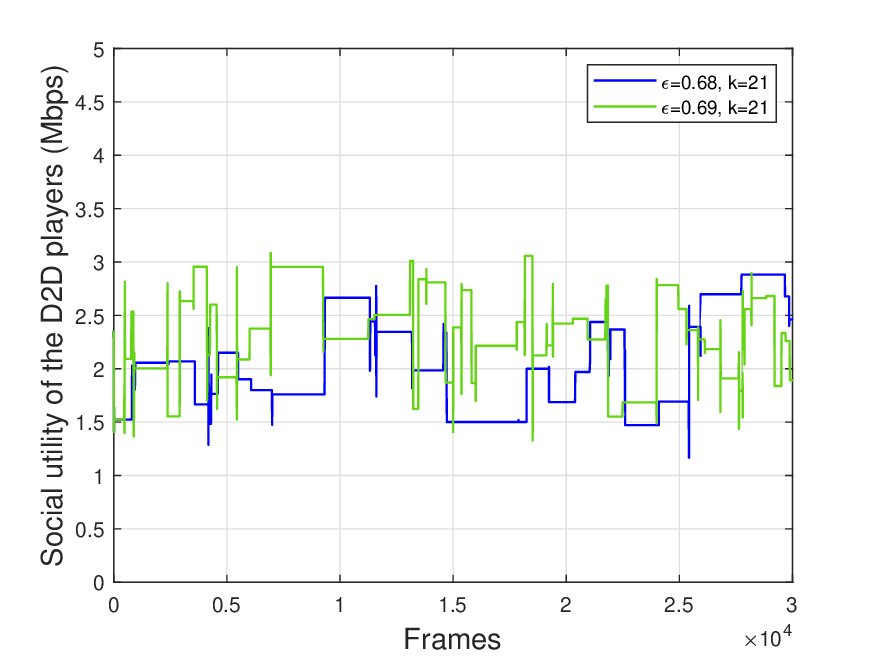}}
 	\caption{Social utility of the D2D players over frames with Algorithm 1.}\label{fig:2}
 	\end{figure}
 	
\subsection{With Shadowing, Fast Fading and Mobility of CUs}
We increase  the number of D2D pairs $N_D$ and the number of CUs $N_C$ to $20$, $\gamma^{tgt}$ to $15$ dB and the minimum required utility $r_d^{tgt}$ for each D2D player to $200$ kbps, normalized by the normalization factor $\alpha$.  The normalization factor $\alpha$ is set to $100$ Mbps. When $N_D$,  $\gamma^{tgt}$ and $r_d^{tgt}$  are increased more than this, we observe in this case too that the social utility of the D2D players fluctuates continuously.\\
\emph{1) Shadow Fading Model:}  We model the shadow fading random variables (RVs) corresponding to the propagation paths with lognormal RVs of standard deviation $8$ dB. We assume these to be constant over time except: 1) the shadow fading RV corresponding to the path between a CU and the BS and 2) the shadow fading RV corresponding to the path between a CU and the D2D receiver which change because of CU mobility as follows. According to the exponential autocorrelation model of shadowing \cite{y}, the shadowing samples are correlated and do not change much over a certain range of distances. Since the CUs are moving, it is difficult to model the shadowing process and its effect on the network. Therefore, we assume that the shadowing RVs remain constant over the decorrelation distance and change after a CU has moved a distance greater than the decorrelation distance \cite{z}. 
	\begin{figure}[t]
 	\centering
 	\scalebox{0.6}{\includegraphics{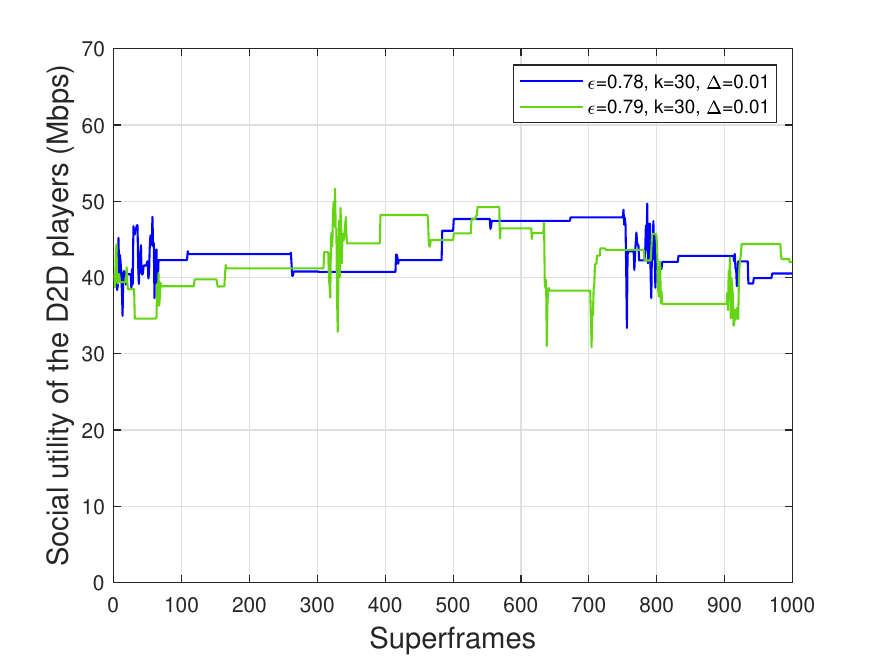}}
 	\caption{Social utility of the D2D players over superframes with Algorithm 3.}
 	\end{figure}
Thus, the shadowing RVs change every decorrrelation distance. The decorrelation distance in the urban macro cell scenario is $50$ m \cite{y}. For a speed of $0.5$ m/s, the shadowing RVs do not change for $100$ s. Since a subframe is of duration $1$ ms, the number of subframes in which the shadowing RVs do not change is $10^5$. \\
\emph{2) Fast Fading Channel Model:} We consider a  Rayleigh fast fading channel  in which the channel gains are independent and exponentially distributed RVs with mean $1$. These channel gains change in every subframe.\\
\emph{3) Mobility Model of CUs:} We employ the following model to characterize the mobility of the CUs. The CUs start from their initial positions and choose directions that are uniformly distributed between $0$ and $360$ degrees. The speeds of all the CUs are constant. The number of time slots (subframes) $T$ for which it travels till it changes its direction is a geometrically distributed RV with parameter $p$, where $p$ is the probability of success in each trial. The mean of the geometric RV is therefore $(1- p)/p$. We set $p=10^{-5}$. The CU after $T$ subframes does not pause unlike the random direction model \cite{z1}, instead it again chooses a direction that is randomly uniformly distributed. Its speed remains the same. When it hits the boundary of the circular cell, it chooses a random direction. The range of angles that it can choose from is limited because it should not cross the boundary. If the time when it hits the boundary is a fraction in milliseconds then the decimal part needs to be truncated because the subframes are of 1 ms duration. Thus, the resulting distance is within the boundary. This model differs from the random direction model  because the pause times are zero and the speed is constant instead of a uniformly distributed RV. We assume a constant pedestrian speed of $0.5$ m/s for all the CUs.
\\
\indent  Let $\Delta= 0.01$. With Algorithm 3, we observe the following results. Fig. 6 shows a typical sample path realization of the social utility of the D2D players over superframes with $\epsilon$ and $k$ equal to $0.78$ and $30$ respectively. We increase $\epsilon$ to $0.79$ and observe that the D2D players explore more. From Fig. 6 we notice that the average social utility of the D2D players obtained with Algorithm 3 is about $45$ Mbps. Moreover,  the social utility of the D2D players  remains above $30$ Mbps over the superframes which means that the algorithm is robust even with shadowing, fast fading and mobility of CUs. We also note that this algorithm can be deployed for local data transmissions in D2D communications underlaying an LTE network because each D2D user gets an average social utility greater than 1.5 Mbps. \\

\section{Conclusions}
We have proposed two semi-distributed suboptimal D2D resource allocation algorithms in a game theoretic framework without inter-D2D interference. Since the algorithms are semi-distributed, DCSI-R doesn't need to be transmitted to the BS which reduces the control overhead and power wastage of the network. We explain the notions of stochastically stable states and Pareto optimality. Our first  proposed algorithm does not perform well because the game's interdependence cannot be ascertained in some frames which shows the limitations of the proposed design. We modify it and propose another suboptimal algorithm that shows sturdy performance to channel randomness and mobility of CUs. It can be deployed for local data transmissions in D2D communications underlaying an LTE network. As very less research has been done in the field of designing Pareto optimal D2D resource allocation algorithms, our work will help advance this field. 

\section*{Acknowledgment}
The author would like to thank Prof.  Abhay Karandikar and Prof. Prasanna Chaporkar for their invaluable advise and suggestions while doing this research work.


\addcontentsline{toc}{chapter}{Bibliography}
\begin{thebibliography}{12}
\bibitem{a0}N. H. Hussein, C. T. Yaw, S. P. Koh, S. K. Tiong and K. H. Chong, ``A Comprehensive Survey on Vehicular Networking: Communications, Applications, Challenges, and Upcoming Research Directions," \emph{IEEE Access}, vol. 10, pp. 86127-86180, Aug. 2022.

\bibitem{a1}5GAA, ``C-V2X in action - 5GAA", https://5gaa.org/c-v2x in action/, accessed: 2023-05-25. 

\bibitem{a2}3GPP TR 26.985 V16.0.0, ``5G; Vehicle-to-everything (V2X); Media handling and interaction", 2010. 

\bibitem{a3}R. Bassoli, F. H. P. Fitzek and E. C. Strinati, ``Why do we need 6G?,'' \emph{ITU J-FET}, vol. 2, no. 6, pp. 1-31, Sept. 2021.  

\bibitem{a4}M. S. M. Gismalla, A. I. Azmi, M. R. B. Salim, M. F. L. Abdullah, F. Iqbal, W.A. Mabrouk et al., ``Survey on Device to Device (D2D) Communication for 5GB/6G Networks: Concept, Applications, Challenges, and Future Directions," \emph{IEEE Access}, vol. 10, pp. 30792-30821, Mar. 2011. 

\bibitem{a}J. Wang, D. Zhu, C. Zhao, J. C. F. Li and M. Lei, ``Resource Sharing of Underlaying Device-to-Device and Uplink Cellular Communications,'' \emph{IEEE Commun. Lett.}, vol. 17, no. 6, pp. 1148-1151, June 2013.

\bibitem{b}D. Feng, L. Lu, Y. Y. Wu, G. Y. Li, G. Feng and S. Li, ``Device-to-Device Communications Underlaying Cellular Networks,'' \emph{IEEE Trans. Commun.}, vol. 61, no. 8, pp. 3541-3551, Aug. 2013.

\bibitem{c}W. Zhao and S. Wang, ``Resource Allocation for Device-to-Device Communication Underlaying Cellular Networks: An Alternating Optimization Method,'' \emph{IEEE Commun. Lett.}, vol. 19, no. 8, pp. 1398-1401, Aug. 2015.

\bibitem{d}S. Sesia, I. Toufik and M. Baker, ``\emph{LTE-The UMTS Long Term Evolution: From Theory to Practice},'' 2nd ed., New York, NY, USA: Wiley, 2011.

\bibitem{e}J. R. Marden, H. P. Young and L. Y. Pao, ``Achieving Pareto Optimality through Distributed Learning,'' in \emph{Proc.} CDC, pp. 7419-7424, Dec. 2012.  

\bibitem{f}A. Menon and J. S. Baras, ``A Distributed Learning Algorithm with Bit-valued Communications for Multi-agent Welfare Optimization,'' in \emph{Proc.} CDC, pp. 2406-2411, Dec. 2013. 

\bibitem{g1}P. Chaporkar, A. Proutiere and B. Radunoviae, ``Rate Adaptation Games in Wireless LANs: Nash Equilibrium and Price of Anarchy,'' in  \emph{Proc.} INFOCOM, pp. 1-9, Mar. 2010. 

\bibitem{g}T. Hatanaka, Y. Wasa, R. Funada, A.G. Charalambides and M. Fujita, ``A Payoff-Based Learning Approach to Cooperative Environmental Monitoring for PTZ Visual Sensor Networks'' \emph{IEEE Trans. Automatic Control}, vol. 61, no. 3, pp. 709-724, Mar. 2016.

\bibitem{h}H. P. Borowski, J. R. Marden and J. S. Shamma, ``Learning efficient correlated equilibria,''in \emph{Proc.} CDC, pp. 6836-6841, Dec. 2014.

\bibitem{i}J. R. Marden, ``Selecting Efficient Correlated Equilibria Through Distributed Learning,'' \emph{Games Econ. Behav.}, vol. 106, issue C, pp. 114-133, Nov. 2017.

\bibitem{j}F. Wang, L. Song, Z. Han, Q. Zhao and X. Wang, ``Joint Scheduling and Resource Allocation for Device-to-Device Underlay Communication,'' in \emph{Proc.} WCNC, pp. 134-139, Apr. 2013.

\bibitem{k}X. Chen, R. Q. Hu and Y. Qian, ``Distributed Resource and Power Allocation for Device-to-Device Communications Underlaying Cellular Network,'' in \emph{Proc.} GLOBECOM, pp. 4947-4952, Dec. 2014.

\bibitem{l}R. Yin, C. Zhong, G. Yu, Z. Zhang, K. K. Wong and X. Chen, ``Joint Spectrum and Power Allocation for D2D Communications Underlaying Cellular Networks,'' \emph{IEEE Trans. Veh. Tech.}, vol. 65, no. 4, pp. 2182-2195, Apr. 2016.

\bibitem{m}C. X. L. Song, Z. Han, Q. Zhao, X. Wang, X. Cheng and B. Jiao, ``Efficiency Resource Allocation for Device-to-Device Underlay Communication Systems: A Reverse Iterative Combinatorial Auction Based Approach,'' \emph{IEEE J. Sel. Areas. Commun.}, vol. 31, no. 9, pp. 348-358, Sept. 2013.

\bibitem{n}S. Maghsudi and S. Stanczak, ``Hybrid Centralized-Distributed Resource Allocation for Device-to-Device Communication Underlaying Cellular Networks,'' \emph{IEEE Trans. Veh. Tech.}, vol. 65, no. 4, pp. 2481-2495, Apr. 2016.

\bibitem{o}S. Lindner, R. Elsner, P. N. Tran and A. Timm-Giel, ``A Two-Game Algorithm for Device-to-Device Resource Allocation with Frequency Reuse,'' in \emph{Proc.} VTC2019-Fall, pp. 1-5, Sept. 2019.

\bibitem{p}Y. Sun, F. Wang and Z. Liu, ``Coalition Formation Game for Resource Allocation in D2D Uplink Underlaying Cellular Networks,' \emph{IEEE Commun. Lett.}, vol. 23, no. 5, pp. 888-891, May 2019.

\bibitem{q}C. Sun, B. Zhang, W. Zheng and J. Shu, ``Price Learning in Joint Resource Allocation and Power Control Game in Heterogeneous Cellular Networks,'' in \emph{Proc.} ISPA/BDCloud/SocialCom/SustainCom, pp. 1354-1358, Dec. 2020.

\bibitem{r}D. Shi, L. Li, T. Ohtsuki, M. Pan, Z. Han and H. V. Poor, ``Make Smart Decisions Faster: Deciding D2D Resource Allocation via Stackelberg Game Guided Multi-Agent Deep Reinforcement Learning,'' ," in \emph{IEEE Trans. Mobile Comput.}, vol. 21, no. 12, pp. 4426-4438, Dec. 2022.

\bibitem{s}Y. Yuan, Z. Li, Z. Liu, Y. Yang and X. Guan, ``Double Deep Q-Network Based Distributed Resource
Matching Algorithm for D2D Communication,'' \emph{IEEE Trans. Veh. Tech.}, vol. 71, no. 1, pp. 984-993, Jan. 2022.

\bibitem{t}M. Singh and P. Chaporkar, ``An Efficient and Decentralised User Association Scheme for Multiple Technology Networks,'' in \emph{Proc.} WiOpt, pp. 460-467, May 2013.

\bibitem{u}J. R. Marden, S. D. Ruben and L. Y. Pao, ``A Model-Free Approach to Wind Farm Control Using Game Theoretic Methods,'' \emph{IEEE Trans. Control Sys. Tech.} vol. 21, no. 4, pp. 1207-1214, July 2013.

\bibitem{v}H. Dai, Y. Huang, R. Zhao, J. Wang and L. Yang, ``Resource Optimization for Device-to-Device and Small Cell Uplink Communications Underlaying Cellular Networks,'' \emph{IEEE Trans. Veh. Tech.}, vol. 67, no. 2, pp. 1187-1201, Feb. 2018.

\bibitem{w}S. Dominic and L. Jacob, ``Utility-Based Resource Allocation for Underlay D2D Networks,''in \emph{Proc.} TENSYMP, pp. 1-5, July 2017.

\bibitem{x}I. Mondal, A. Neogi, P. Chaporkar and A. Karandikar, ``Bipartite Graph Based Proportional Fair Resource Allocation for D2D Communication,''  in \emph{Proc.} WCNC, pp. 1-6, Mar. 2017.
  

\bibitem{y}``Guidelines for evaluation of radio interface technologies for IMT-Advanced,'' \textit{Rep. ITU-R M.2135}, 2008.

\bibitem{z}T. Abbas, K. Sj\"{o}berg, J. Karedal and F. Tufvesson, ``A Measurement Based Shadow Fading Model for Vehicle-to-Vehicle Network Simulations,'' \textit{Int. J. Antennas Propag.}, vol. 2015, Article ID 190607, June 2015. 

\bibitem{z1}Y-B. Ko and N. H. Vaidya, ``Location Aided Routing (LAR) in mobile ad hoc networks,'' \textit{Wireless Networks}, vol. 6, no. 4, pp. 307–321, Sept. 2000.
\end{thebibliography}
\end{document}